\documentclass[12pt]{iopart}

\usepackage[colorinlistoftodos]{todonotes}

\expandafter\let\csname equation*\endcsname\relax
\expandafter\let\csname endequation*\endcsname\relax
\usepackage{amsmath}
\usepackage{amssymb}
\usepackage{tikz-cd}
\usepackage{graphicx}

\usepackage[justification=centering]{caption}

\usepackage{placeins}

\usepackage{verbatim}

\usepackage{array}

\newcolumntype{L}[1]{>{\raggedright\let\newline\\\arraybackslash\hspace{0pt}}m{#1}}
\newcolumntype{C}[1]{>{\centering\let\newline\\\arraybackslash\hspace{0pt}}m{#1}}
\newcolumntype{R}[1]{>{\raggedleft\let\newline\\\arraybackslash\hspace{0pt}}m{#1}}

\newenvironment{definition}[1][Definition]{\begin{trivlist}
\item[\hskip \labelsep {\bfseries #1}]}{\end{trivlist}}
\newenvironment{proof}[1][Proof]{\begin{trivlist}
\item[\hskip \labelsep {\bfseries #1}]}{\end{trivlist}}

\newtheorem{theorem}{Theorem}[section]

\newcommand{\la}{\langle}
\newcommand{\ra}{\rangle}
\newcommand{\lb}{\left(}
\newcommand{\rb}{\right)} 

\begin{document}

\title[Quantum Control Landscapes are Almost Always Trap Free]{Quantum Control Landscapes are Almost Always Trap Free: A Geometric Assessment}

\author{Benjamin Russell}
\address{Department of Chemistry, Princeton University}
\ead{br6@princeton.edu}

\author{Herschel Rabitz}
\address{Department of Chemistry, Princeton University}
\ead{hrabitz@princeton.edu}

\author{Re-Bing Wu}
\address{Department of Automation, Tsinghua University}
\ead{rbwu@tsinghua.edu.cn}

\date{\today}

\begin{abstract}
A proof is presented that almost all closed, finite dimensional quantum systems have trap free (i.e., free from local optima) landscapes for a large and physically general class of circumstances.
This result offers an explanation for why gradient-based methods succeed so frequently in quantum control.
The role of singular controls is analyzed using geometric tools in the case of the control of the propagator, and thus in the case of observables as well.
Singular controls have been implicated as a source of landscape traps.
The conditions under which singular controls can introduce traps, and thus interrupt the progress of a control optimization, are discussed and a geometrical characterization of the issue is presented.
It is shown that a control being singular is not sufficient to cause control optimization progress to halt, and sufficient conditions for a trap free landscape are presented.
It is further shown that the local surjectivity (full rank) assumption of landscape analysis can be refined to the condition that the end-point map is transverse to each of the level sets of the fidelity function.
This mild condition is shown to be sufficient for a quantum system's landscape to be trap free.
The control landscape is shown to be trap free for all but a null set of Hamiltonians using a geometric technique based on the parametric transversality theorem.
Numerical evidence confirming this analysis is also presented.
This new result is the analogue of the work of Altifini, wherein it was shown that controllability holds for all but a null set of quantum systems in the dipole approximation.
These collective results indicate that the availability of adequate control resources remains the most physically relevant issue for achieving high fidelity control performance while also avoiding landscape traps.
\end{abstract}

\maketitle

\section{Introduction: Control Landscape Analysis}

The study of quantum control landscapes (i.e., specific objective cost that depends on the of quantum time evolution operator as a function of an external field) is a topic of prime interest for assessing the viability of reaching a desired control outcome.
As background, prior work has focused on applying differential geometry to several issues in quantum optimal control and quantum mechanics generally \cite{me1, me2, rage, BroA, BroB, BroC, ACAR, GQM, NFINQSL}.
Other recent work \cite{alt} has also focused on applying geometric and Lie algebraic methods to controllability in quantum systems.
Following similar principles, in this paper geometric methods are brought to bear on the analysis of quantum control landscapes.
Specifically, we apply the geometric understanding of the `size' (\emph{measure}) of the set of singular values of smooth maps (Sard's theorem \cite{Sard}) to quantum control landscapes.
Further, we apply the parametric transversality theorem, (\cite{singbook} lemma 1, chapter 2) to reach the main conclusion.
In particular, the present work facilitates understanding of the wide prevalence of trap free quantum control landscapes seen in practice, while also accommodating the fact that quantum control landscapes possessing traps appear highly atypical.

There has been success, in both experiment \cite{Soa, num2, num1} and extensive simulation \cite{sqcs, sevcon, kmev1, kmev2} with the application of gradient based (or other local gradient estimation) algorithms in quantum control.
The gradient ascent algorithm is sensitive the critical topology of the function being optimized.
In particular, gradient ascent can converge to a local optimum if started from specific initial conditions within the basin of attraction.
Consideration of this prospect garnered some controversy.
In this regard, some potential issues have also been identified \cite{sscl, atql} and debated \cite{coarql} which may affect convergence to a global optima on the control landscape for specific systems.
It has been shown that the presence of singular controls \cite{Wu, sqcs, sscl, atql, coarql} can encumber the gradient ascent procedure in some very specific systems, all of which had to be specially constructed to have this unfavorable property.
For the definition of singular controls, see sec. (\ref{singsec}) and the appendix (\ref{singapp}).
Particularly the case of constant, or even zero, control fields has received attention \cite{sscl, atql}, partly because of the tractability of analyzing this case.
Some examples of non-constant singular controls have been found numerically in \cite{Wu}, but they were all found to be saddles rather than local optima.
There is also mounting numerical evidence that situations where singular controls inhibit progress during the gradient ascent procedure are, in several senses, rare \cite{Wu, Riv, an1, num2}.
This collected evidence for the common ease of quantum control optimization, and the evident rarity of landscape traps motivated the pursuit of the central result of the present paper.

Finally we note that this work does not address the \emph{rate} of convergence to the global optima and that this issue is also dependent on the particular search algorithm employed.
Some numerical efforts demonstrated that a favorable rate of convergence is typically the case and remains so as the number of levels rises \cite{searcheff, kmev1, kmev2}.

\subsection{Summary of The Central Theorem Established in this Paper}

Here we first summarize the key finding in this paper, followed by a detailed proof and discussion of the result.
In particular we show that only a null set of quantum systems (within the space of all systems with a given number of levels) possess traps caused by singular critical controls.
The sense in which this set is null can be understood as follows: \emph{if quantum systems are generated at random} there is zero probability of finding an example with any singular traps.
For a rigorous introduction to the analysis of measures and null sets, see \cite{fran}.

The crux of the proof rests on a novel application of the parametric transverality theorem from differential geometry \cite{hirsch}.
This result has significant implications for quantum control, which can be expressed informally by the following statement:

\vspace{0.5cm}
\fbox{\parbox{0.9\textwidth}{
Consider a parameterized family of time dependent Hamiltonians $H[E_n, \lambda_m]$ which depend on finitely many control variables $E_n$ and finitely many real additional parameters $\lambda_m$.
If, for all $E_n$ and $\lambda_m$ (other than those corresponding to the global optima of the fidelity function), it is possible to increase fidelity by applying variations $\delta E_n, \delta \lambda_m$ to $E_n$ and $\lambda_m$ respectively, then it follows that the landscape as a function only of $E_n$, is trap free for \emph{almost all} fixed values of $\lambda_m$ (i.e., for all but a null set of $\lambda_m$ values).
}}
\vspace{1cm}

This informal statement assumes that there is a physically suitable fidelity function quantifying the `cost' of a Hamiltonian driven time evolution.
For the above statement to hold, it is also required that the chosen fidelity function does not possess traps of its own, i.e. built into its mathematical structure.
This circumstance is the case for many fidelity functions in popular use \cite{cjw}.
An example of a parameterized family of quantum systems, which could be assessed using the above boxed statement is given below:
\begin{align}
  H[E_n, \lambda_m](t) & =
  \begin{pmatrix}
	\lambda_0 & \lambda_1 \\ \lambda_1 & -\lambda_0
  \end{pmatrix}  
  + E_0 \sin(t)
  \begin{pmatrix}
	0 & 1 \\ 1 & 0
  \end{pmatrix} \\ \nonumber
  & + E_1 \cos(t)
  \begin{pmatrix}
	0 & i \\ -i & 0
  \end{pmatrix}
  + \left(E_2 \sin(2t) + E_3 \cos(3t)\right)
  \begin{pmatrix}
	1 & 0 \\ 0 & -1
  \end{pmatrix}
\end{align}
The key implication of the results in this work is that gradient ascent methods will almost always succeed (independent of any initial guess for a control) for optimizing quantum dynamics (i.e., discovering a pulse of maximal fidelity) for almost all quantum systems.
This conclusion applies to quantum systems with any finite number of levels.
The result has practical significance for finding optimal control fields in simulation as well as in learning control experiments \cite{in2} attempting to discover shaped pulses which maximize fidelity in the laboratory.
Such experiments, whether guided with a gradient procedure or an other suitable method are essentially a laboratory realization of the analogous simulation, wherein, a control pulse is systematically updated using the algorithmic rule until a high fidelity pulse is found.
Henceforth in this work, the distinction between a control function $E(t)$, and vector of control variables $E_m$ will be suppressed unless the distinction is vital.

No attempt is made in this work to address the size of the basin of attraction of local traps in any given quantum control landscape.
This interesting and important question is assessed in \cite{sqcs} where it is found that the attractor basin of traps in a few specific examples are small in an appropriate measured sense.

\subsection{Quantum Systems and the Goals of Control}

This paper studies the control of the quantum propagator for finite level quantum systems with Hamiltonians of the form:
\begin{align}
	\label{schcon}
	\hat{H}(t) = \hat{H}_0 + E(t) \hat{H}_c
\end{align}
where both the drift $H_0$ and control $H_c$ Hamiltonians are, respectively, traceless and $E(t)$ is drawn from a finite, but potentially very high dimensional parameterized space of time dependent control fields.
The traceless condition is taken in order that only evolution operators $U_t \in SU(n)$ need to be considered rather than $U(n)$, which contains information about a physically redundant overall phase.
For a more detailed discussion about the distinction between $U(n)$ and $SU(n)$ with respect to singular controls and the effect of this specialization on the control landscape see \cite{sscl, jdom}.

Initially this work studies the \emph{maximally} controlled system, for which every Hamiltonian matrix element considered as under control; in this circumstance it is reasonably shown that the landscape is trap free.
Importantly, incrementally reducing this latter high degree of control back down to the form of (\ref{schcon}) is shown to introduce \emph{no} traps, unless the reduction procedure is engineered to do so by fixing particular matrix elements to specific values within a set shown to be null at each step of the procedure.

In the case of controlling the propagator of a quantum system, the fidelity can be represented by a functional $F_1$ to be optimized:
\begin{align}
	\label{fun_1}
	F_1[E] = \Re\left\{Tr\left(\hat{G}^{\dagger}\hat{U}_T\right) \right\}
\end{align}
where $T$ is some fixed final time and $G$ a target unitary.
A phase independent version of this function is given by:
\begin{align}
	\label{fun_2}
	F_2[E] = \left|Tr\left(\hat{G}^{\dagger}\hat{U}_T\right)\right|^2
\end{align}
For several applications of this functional and further discussion see \cite{sscl, jdom}.
These functionals both have the form:
\begin{align}
F[E] = J(V_T[E])
\end{align}
where $J:SU(n) \rightarrow \mathbb{R}$ is defined as $J(U) = \left|Tr\left(\hat{G}^{\dagger}\hat{U}\right)\right|^2$ (for example) and $V_T$ is the end-point map, that assigns to a control $E$ to the final time propagator $V_T[E]=U_T$ which solves the corresponding Sch\"ordinger equation at time $T$.
The overall conclusion of the paper applies equally to both cost functions $F_1, F_2$ which have been shown to both possess only global and local maxima and saddle points \cite{jdom}.

Throughout this paper $T$ will be assumed to be large enough for the system to be fixed time controllable, that is to say, for any end-point unitary $U_T$, there is a control $E$ which implements it in time $T$.
This property is also referred to as \emph{accessibility}, it compliments the property of straightforward controllability, which states that: for each end-point unitary $U$, the exists a control $E$ which implements $U$ at \emph{some} final time $T$.
i.e., there exists at least one control $E$ such that $U_T = G$ for any given $G$.
A proof of the existence of such a time $T$, which includes our present case, is given in (\cite{jur} theorem 30, Ch. 3).

There are several common objectives in quantum control.
These can be broadly classified physically as:

\begin{enumerate}
	\item \label{prob1} Maximizing the expectation of a given (Hermitian) observable $\hat{O}$ at time $T$
	\item \label{prob2} Controlling an initial quantum state $|\psi_{I} \rangle$ to reach a desired quantum state $|\psi_{F} \rangle$ at time $T$
	\item \label{prob3} Controlling the quantum propagator $\hat{U}_T$ so that it is driven to a desired goal $\hat{G} \in SU(n)$.
In quantum information sciences applications $\hat{G}$ is typically interpreted as implementing a quantum gate.
\end{enumerate}
Although we will focus on the third of these objectives, the conclusions in this work apply to all of the above tasks as the objective function in the cases (\ref{prob1}), (\ref{prob2}) can be expressed as a function of the quantum propagator $U_T$.
For a discussion of the correspondence between tasks (\ref{prob2}) and (\ref{prob3}) from a geometric perspective, see appendix (\ref{geocon}).

\subsection{The Gradient Algorithm and the Prospect of Encountering Traps}

Employment of the gradient algorithm to determine a control that maximizes fidelity naturally requires the computation of the gradient of the functional $F$.
The gradient of the functionals (\ref{fun_1}), (\ref{fun_2}) can be computed in closed form \cite{cqp}.
The variation of the end point w.r.t. the full Hamiltonian is given by:
\begin{align}
\delta U_T = U_T \int_0^{T} U_t^{\dagger} i \delta H(t) U_t dt
\end{align}
For more complex forms of control coupling, the gradient is not as simple when non-linear coupling to a control field is included, as in the case of control via polarizability \cite{mepol}.
However, this case is essentially the same as the first variation of the end point map with respect to the control in  the linear coupling case, except, $\delta H(t)$ takes a different form.
For further discussion of gradient decent/ascent methods see \cite{grad}.
Employment of the gradient algorithm is important for understanding the topology of both the set of critical points, and that of the local (if any exist) and global maxima, as exemplified in figure (\ref{gfig}).
\begin{figure}[!ht]
  \centering
    \reflectbox{
 \includegraphics[width=0.25\textwidth]{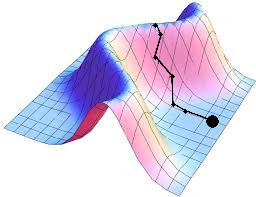}
 \includegraphics[width=0.4\textwidth]{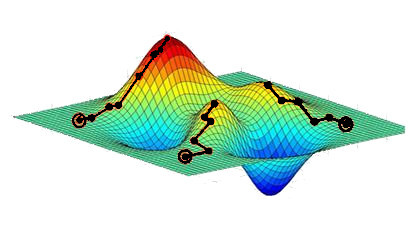}}
  \caption{Gradient algorithms can converge to local optima, or \emph{traps}, if they exist (left), but will reach a global maximum point if no traps exist (right).}
\label{gfig}
\end{figure}
Here we fix terminology relating to traps in the control landscape of a quantum system.
\begin{definition}
A control $E$ is called a \emph{critical control} or a \emph{critical point} w.r.t. a given $F=J \circ V_T$, if $\frac{\delta F}{\delta E} = 0$.
\end{definition}
\begin{definition}
A control $E$ is called a \emph{a second order critical point} w.r.t. a given $F=J \circ V_T$ if $\frac{\delta F}{\delta E} = 0$ (i.e., a critical control) and the Hessian $\frac{\delta^2 F}{\delta E^2}$ is negative semi-definite.
Such an $E$ may or may not be a true local optimum (i.e., a trap) depending on the nature of the higher derivatives with respect to $E$.
\end{definition}
\begin{definition}
A control $E$ is called a \emph{trap} w.r.t. a given $F=J \circ V_T$ if it is a local, but not global optima of the same $F$.
\end{definition}

\subsection{The Three Assumptions of Landscape Analysis}

There are three key assumptions of landscape analysis (see table \ref{axiomstab}), which are known to be sufficient for a quantum control landscape to be trap free.
A clear statement, both in control theoretic and differential geometric terms, can be found in section (\ref{axioms}).
A key result in much of the recent work in landscape analysis is that these three assumptions imply that the gradient algorithm will converge to a globally optimal control without getting `stuck' in a local optima (i.e., a trap).
There is numerical evidence \cite{kmev1, kmev2} that this hypothesis holds with a wide scope of validity as well as experimental evidence \cite{num2, num1, sun1, sun2}.

Earlier important work \cite{alt} has shown that the controllability assumption (1) in Table (\ref{axiomstab}) discussed later is almost always satisfied given a pair $(H_0, H_c)$.
That is, the set of Hamiltonians for which controllability fails is a null set.
It is however, an open problem to completely characterize the set on which controllability fails.
The two most widely applied and useful criteria for controllability are the Lie algebra rank condition (LARC) and the connected graph criterion \cite{alt}.
A clear example for numerically checking controllability by the Lie algebra rank condition can be found in \cite{db} where an insightful graphical representation of the process is presented.
However, only limited work has been performed on the the satisfaction of the local surjectivity assumption (\ref{singsec}) in Table (\ref{axiomstab}) and its impact on the performance of gradient algorithms.
In this work we present an analogous result to that of \cite{alt} for controllability,  which applies to the local surjectivity assumption (\ref{singsec}).
In addition, the final assumption (3), is that there are sufficient control resources (i.e., the control space is sufficient for the end-point map to be globally surjective) to freely explore the landscape.
The new result, considering all three assumptions, explains why gradient ascent convergence to a global optimum is so typical in practice despite the fact that some engineered special examples with traps are known.

\section{Singular Controls and Properties of the End-Point Map}

\label{singsec}

\subsection{Kinematic and Dynamical Optima}

The scenario here concerns discovering a control scheme driving the time evolution operator $U_t$ to a desired goal $G$ at time $T$ (i.e., $U_T = G$).
This process is represented by the following commutative diagram,
\begin{figure}[h!]
	\begin{center}
		\begin{tikzcd}
		    C \arrow[]{r}{V_T} \arrow[bend right=30, swap]{rr}{F = J \circ V_T} & SU(n) \arrow[]{r}{J} & {[0,1]}
	    \end{tikzcd} 
	\end{center}
\caption{The chain rule connects the dynamical $F$ and kinematic $J$ landscape}
\label{chainrule}
\end{figure}
$J:SU(n) \rightarrow [0,1]$ is an objective function to maximize, and $C$ is a pre-defined space of control fields.
Applying the chain rule to $F=J \circ V_T$ yields:
\begin{align}
	\label{chain}
	\frac{\delta F}{\delta E} = \frac{d J}{d V_T[E]} \circ \frac{\delta{V_T}[E]}{\delta E}
\end{align}
\begin{definition}
A control $E$ is said to be \emph{singular} if the set of all $\delta V_T$ for which there exists a corresponding $\delta E$ (i.e., a value of $\delta V_T[\delta E]$) doesn't span $T_{V_T[E]}SU(n)$ at the point $V_T[E]$.
This is to say that the Fr\'echet derivative:
\begin{align}
\delta V_T \big|_{E} : \delta E \mapsto \delta U_T
\end{align}
is not of maximal rank at the point $E$ in the space of controls.
The co-dimension of the image of $dV_T \big|_{E}$ given by $\text{dim}(SU(n)) - \text{rank}(dV_{T} \big|_{E})$ is called the \emph{co-rank} of the control.
\end{definition}
There are two types of critical points of $F$: ones for which $\frac{d J}{d V_t} = 0$, and those for which $\frac{\delta{V_T}}{\delta E}$ is not of full rank in $T_{V_T[E]}SU(n)$.
Consideration of $J$ is referred to as the kinematic control landscape and $F$ is referred to as the dynamic control landscape.
Figure (\ref{chainrule}) clarifies the relationship between the kinematical and dynamical landscapes.
Ultimately, the goal is to understand if singularities of $V_T$ can introduce critical points of $F$ which are not critical points of $J$, i.e. singular controls which introduce new critical points into the landscape of $F$.
Understanding which systems have no such singular critical controls will elucidate the circumstances for which the critical point structure of $F$ and $J$ are the same.
For these systems an analysis of the kinematic landscape alone suffices to understand the full control landscape of $F$.
This is a desirable goal as it facilitates reaching the conclusion that one only needs to consider the prospect of traps arising directly from a fidelity function, and thus that an appropriate choice of fidelity function $J$ is sufficient to result in a trap free landscape \cite{cjw}.

With the remarks above in mind we make the following definitions and observation:
\begin{definition}
	A \emph{Kinematic Critical Point} is a control $E$ such that $\frac{d J}{d V_t[E]} = 0$.
\end{definition}
\begin{definition}
	A \emph{Dynamic Critical Point} is a control $E$ such that $\frac{\delta F}{\delta E}=0$.
\end{definition}
It is clear from eqn. (\ref{chain}) that all kinematic critical points are dynamic but that the converse is not true unless $\frac{ \delta V_T[E]}{\delta E}$ is full rank for all $E \in C$.
It is important to understand the nature of both types of critical points.
Ultimately, the most salient question is: do singular controls introduce local optima into quantum control landscapes and if so, what is the practical ramifications of this both in simulations and laboratory learning control \cite{in2} scenarios?

\subsection{$V_T$ is a smooth map}

Here properties of $V_T$ needed in order to apply the parametric transversality theorem are established.
Considering controls only drawn from either $C^{\infty}([0,T])$ or any finite dimensional vector space of smooth functions, $V_t$ has the properties of being a smooth function of $t$, a smooth function of $E$ and a smooth function of $H_0$ and $H_c$.
We use this fact without giving a proof, however a proof can be given by using the so called `convenient calculus' \cite{concalc}.
The proof is long and involved as well as a direct parallel of many existing proofs, so it is omitted.
We further observe, by the smoothness of the matrix exponential and of matrix multiplication, that the end-point map is also smooth on any space of piecewise smooth controls.
This fact will be required in ensuing geometric analysis.

\section{Climbing the Landscape: Transverality to the Level Sets of $J$ is Sufficient to Climb}

This section will show the failure of local surjectivity will not necessarily cause a gradient assent to halt.
We show that a significantly weaker condition rather than local surjectivity is sufficient to exclude traps, namely: $V_T$ being transverse to the level sets of the fidelity function.
We also argue that the set of Hamiltonians for which a search will halt are a null set under some physically reasonable assumptions.

\subsection{Transversality}

Firstly, the concept of transversality is introduced as an abstract property of smooth maps between manifolds.
Secondly, the end point map is shown to possess the property of transversality by taking specific instances of the manifolds in the definition of transversality.

The gradient of a smooth function on a smooth manifold is always perpendicular to level sets of this function (if the same Riemannian metric is used throughout).
If local surjectivity of $V_T$ fails somewhere on a specific level set of $J$ on $SU(n)$, it may not matter as far as climbing the landscape when using gradient assent. All that matters is an ability to ascend the landscape, not necessarily to traverse the level set itself.
If there does not exist a $\delta E$ which causes $V_T[E] = U_T$ to vary in a specific direction within the level set of $J$ containing $V_T$, this is not problematic for gradient ascent, but would still indicate the failure of surjectivity.
A control may be singular, even up to co-rank of the dimension of the level sets of $J$ containing $V_T[E]$, but as long as there exists a $\delta E$ such that $\delta (V_t[E])$ has a non-zero component in the direction of the gradient $\nabla J$ on $SU(n)$, then it will not impair the ability to climb the landscape (i.e., increase $F$) by introducing a small variation of $\delta E$.
With this in mind, we state the following definition of a transverse map:
\begin{definition}[Transverse map \cite{singbook}], Given two smooth manifolds, $M, N$, a submanifold $L \leq N$, and a  smooth map $\phi: M \rightarrow N$ (with $L \subset \text{Image}(\phi)$), $\phi$ is called \emph{transverse} to $L$ (denoted $\phi \pitchfork L$):
\begin{align}
\text{Image}\left(d \phi \big|_{p} \right) \oplus T_{\phi(p)} L = T_{\phi(p)} M \ \ \forall p \in \phi^{-1}(L)
\end{align}
\end{definition}
The concept above is illustrated in figure (\ref{fig:trans}).
In this work, only the case that $\phi$ is globally surjective (i.e., an onto function) will be important.
This case corresponds to only considering controllable systems.
This renders redundant the condition that $L$ is a subset of the image of $\phi$ because the image of $V_T$ is the whole of $SU(n)$ for quantum systems which are fixed time controllable.
The results reported in \cite{alt} and (Theorem 12, Chap. 6 \cite{jur}), can be paraphrased as: for almost all $a, b \in \mathfrak{su}(n)$, $V_T$ is globally surjective.

\begin{figure}[h]
	\centering
	\includegraphics[scale=0.18]{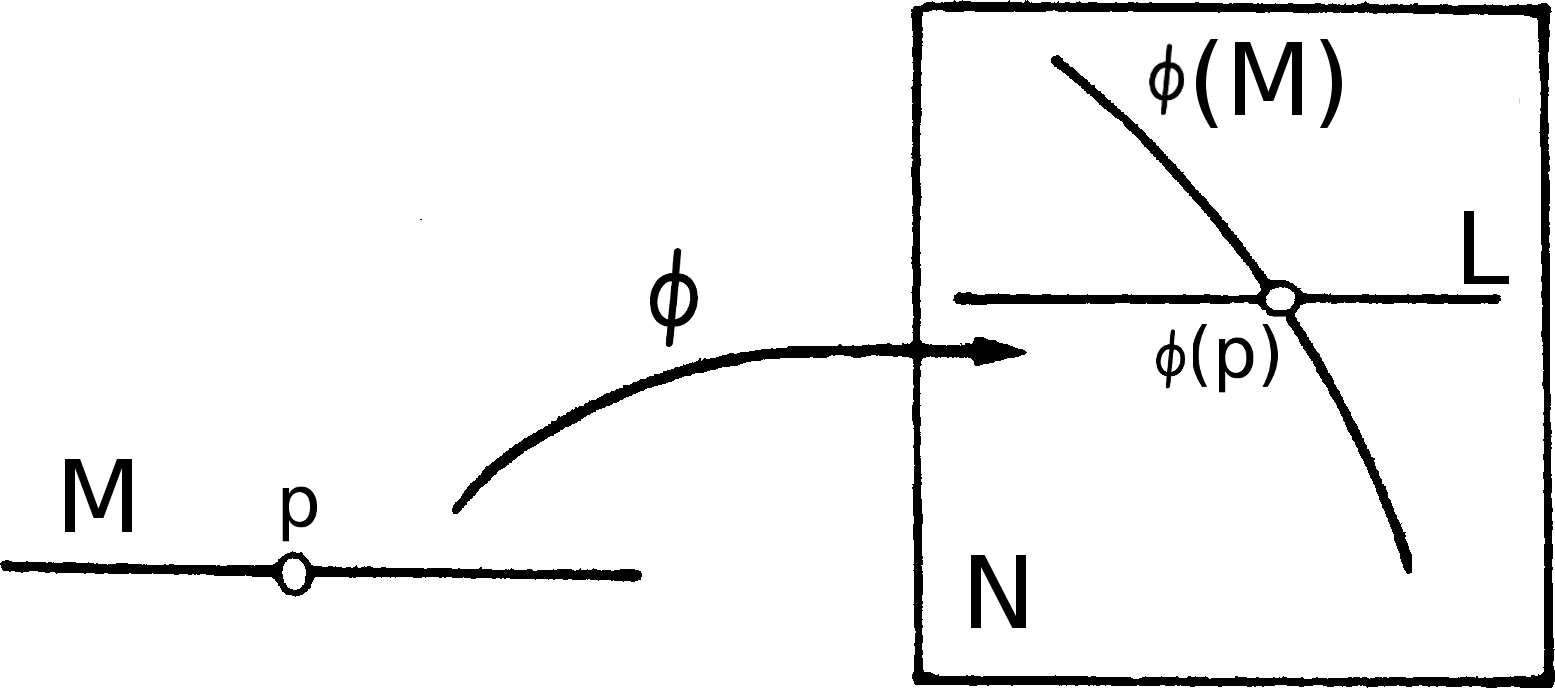}
\caption{Illustration of transversality $\phi \pitchfork L$.
There exists a $\delta p$ such that varying the point $p$ in the domain to $p+\delta p$, the value $\phi(p)$ is varied in a direction $\delta \phi(p)$ steering $\phi(p)$ away from $L$.
Note that in this example $M,L$ are $1$ dimensional and $N$ is $2$ dimensional.
This circumstance, selected for graphical purposes, is in contrast to quantum control applications for the property of $V_T$, where the source manifold (the control space) typically has far higher dimension than the target ($SU(n)$).}
\label{fig:trans}
\end{figure}

\section{Applying the Parametric Transversality Theorem to Quantum Control}

In this section we will utilize the parametric transversality theorem as the key to facilitate proving the rarity of quantum systems with traps.
The proof of this theorem (which follows from Sard's theorem \cite{Sard}) is complex and is omitted, but well known within differential geometry.

Here the statement of the parametric transversality theorem is given prior to it's application to quantum control.
\begin{theorem}[Parametric Transversality Theorem, \cite{singbook} lemma 1, chapter 2] \hfill \break
Given smooth manifolds, $M, N$ and $S$, a submanifold $L \leq N$, and a parameterized family of smooth maps $\phi_s: M \rightarrow N$ where $s \in S$ (parameterized by $s$), then if $\phi:M \times S \rightarrow N$ defined by: $\phi(p,s):=\psi_s(p)$ is transverse to $L$ (when variations of $(p,s)$ are considered), then almost all values of $s \in S$ have $\psi_s$ transverse to $L$.
Furthermore, the pre-image of $L$ is a submanifold of $M$ with codimension in $M$ equal to the codimension of $L$ in $N$.
\end{theorem}

The following result of differential geometry aids in building intuition about Theorem \ref{tcthm}:
\vspace{0.5cm}
\begin{center}
\fbox{\parbox{0.9\textwidth}{
Given a family of smooth maps $\phi_s : M \rightarrow N$ parameterized (by $s$ drawn from a smooth manifold), for almost all values of s, $\phi_s$ is transverse of a given submanifold $L$ if the family as a whole is transverse to the same $L$.
}}
\end{center}
\vspace{0.5cm}
It is also clarifying to note that if $L=N$, then a map being transverse to $L$ is equivalent to it being locally surjective.

The parametric transversality theorem can be applied to the case of quantum control landscapes for finite level systems.
One can set $M$ to be a (finite, large dim, manifold of control fields), $N$ to be $SU(n)$ and $L < SU(n)$ to be any level set of $J$, the cost functional.
In order to analyze the quantum control landscapes, $\phi: M \rightarrow N$ is given by $V_T$ (for some fixed $T$).
The parameter $s \in S$ (the set of maps $V_T$) can be taken to parameterize values in a set of (potentially time dependent) Hamiltonians.
For each time dependent Hamiltonian $iH(t)$, there is an end-point mapping $V_T$.
This is clarified in equation (\ref{maxham}) below and within the ensuing discussion.

\subsection{The Central Theorem}

\label{centhmsec}

In this section we apply the parametric transversality theorem to a large class of quantum control problems.
We show that this analysis allows one to conclude that only a null set of quantum systems have singular critical points.

In this section we denote by $C_{\kappa, p}$, a high but finite, dimensional space of control fields taken to consist of piecewise constant (with $p$ pieces) functions bounded in magnitude by $\kappa$.
We will denote by $\{g_k\}$ a basis for this space, and by $\{B_k\}$ an orthonormal basis of $\mathfrak{su}(n)$.

Consider, initially, the scenario of having \emph{total} control over a system's Hamiltonian $H(t)$ as a function of time (i.e., all matrix elements under control) within the confines of the space $C$, one can then show that successively restricting the degree of control does not generically introduce any singular critical points.
The ensuing physical and mathematical argument in no way rests on the assumption that all such Hamiltonians can be created in the laboratory in practice.
We initially postulate the existence of such a rich space of control fields and the full degree of coupling permitted by these fields, and then progressively restrict the degree of control while assessing the effect that each restriction has on traps in the landscape.
Specifically, we initially study every curve $iH(t) \in \mathfrak{su}(n)$ of the form:
\begin{align}
\label{maxham}
iH(t) = \sum_{j} f_j(t)B_j = \sum_{j,k} a_{j,k} \ g_k(t) B_j
\end{align}
as a Hamiltonian where $g_k$ is given by:
\begin{align}
g_k(t) =
\begin{cases}
1 \ \ \frac{Tk}{p} \leq t \leq \frac{T(k+1)}{p},\\
0 \ \ \text{else}
\end{cases}
\end{align}
The case of a Hamiltonian with drift $H_0$ and control $H_c$, referred to in the single (piecewise constant) field dipole approximation, is included within the above form (\ref{maxham}).
It is now clear that there exists $p,\kappa$ such that the end-point map is surjective as a map from $C_{\kappa,p}$ for some values of $\kappa, p$.
Applying a variation $\delta a_{j,k}$, the end point change becomes:
\begin{align}
U_T^{\dagger} \delta U_T & = \int_{0}^{T} U_t^{\dagger} \delta \left( i H(t) \right) U_t \ dt \\ \nonumber
& = \sum_{j,k} \int_{0}^{T} \delta a_{j,k} \ g_{k}(t) \  U_t^{\dagger} B_j U_t \ dt
\end{align}
Assuming that $\{B_k\}$ is an orthonormal basis of $\mathfrak{su}(n)$, the duration of each piecewise constant segment of the control is short enough (i.e. $p$ is large enough) and the speed of the curve $U_t$ is not too high (i.e. $\kappa$ is small enough), then the set $\left\{ \int_{I_k} U_t^{\dagger} B_j U_t \ dt \right\}$ is also a (not necessarily orthonormal) basis of $\mathfrak{su}(n)$.
Here $I_k$ is the $k^{\text{th}}$ piece defined by $\frac{Tk}{p} \leq t \leq \frac{T(k+1)}{p}$.
To understand exactly what is meant be `short enough', one must understand the singularities of the matrix exponential (\ref{meapp}).
This implies that any variation $\delta U_T$ can be created by some variation $\delta iH(t)$ admissible within this space of piecewise constant controls, and thus that the end-point map is locally surjective everywhere on this space of controls.

Adopting the above premises, one can now apply the parametric transversality theorem to conclude that fixing the value of any one of the control the parameters $a_{j^{'},k^{'}}$ to a given constant will not introduce singular critical points into the landscape for all but a null set of fixed values.

\begin{theorem}
\label{tcthm}
For a system with with Hamiltonian of the form (\ref{maxham}) and space of piecewise constant controls $C_{\kappa, p}$ as above such that the end-point map is locally surjective everywhere in the control space, fixing any single control parameter $a_{j^{'}, k^{'}}=K$ introduces singular critical points into the new control landscape (a function of the remaining, unfixed variables only) only for a null set of values of $K \in \mathbb{R}$.
\end{theorem}
\begin{proof}
Let $V_T[iH(t)]:=U_T$ be the end point map of the Schr\"odinger equation.
This map can be considered to depend on $a_{j,k}$, as these parameters are sufficient to determine $iH(t)$, so we will denote the end-point map $V_T[a_{j,k}]$ accordingly.
As such, $V_T: C_{\kappa,p} \times \mathfrak{su}(n) \rightarrow SU(n)$ (i.e., one control field for each $B_j$) where we have taken $C_{\kappa, p}$ to be sufficient for $V_T$ to be surjective as explained above. 
This is equivalent to saying: $V_T \pitchfork SU(n)$, which in turn implies: $V_T$ is transverse to every sub-manifold of $SU(n)$.
Henceforth, $C$ will have $(\kappa, p)$ left unwritten, assumed to have these values taken so that $V_T$ is locally surjective everywhere in its domain.

We can now decompose the domain of $V_T$ into two parts by selecting one parameter $a_{j^{'}, k^{'}}$, and considering $V_{T}(a_{j^{'}, k^{'}})[a_{j,k}]$ as a one parameter family of maps of the remaining variables $a_{j, k}, \ i\neq i^{'}, \ j \neq j^{'}$.
One can now apply the parametric transversality theorem by considering the values of $a_{j^{'}, k^{'}}$ as $S$, the values of the remaining variables as $M$ and $SU(n)$ as $N$ and conclude that only a null set of values $a_{j^{'}, k^{'}}$ cause this new restricted map to fail to be locally surjective.
$\blacksquare$
\end{proof}

Given theorem (\ref{tcthm}), one can  assume that one value $a_{i_1,j_1}$ has been fixed to a specific value $\kappa$ not in the null set $N_1$ of values which introduce singular critical points.
Fixing another $a_{i_2,j_2}$, and applying an identical argument, one finds that only a null set $a_{i_2,j_2} \in N_2$, introduce singular critical points.
This process can now proceed inductively.
Note that these null sets $N_1, N_2, \ldots$ need not be the same set of values.
Proceeding by induction, one sees that fixing any number of the free parameters describing the system, does not introduce singular critical points except for a null set of fixed values at each step.
It may, however, cause the end-point map to fail to be \emph{globally} surjective, the property generically preserved by fixing parameters is that of being locally surjective within in image of $V_T$, without reference to what that image is.
This eventual breakdown of global surjectivity corresponds to loss of controllability and thus a shrinking of the reachable set.

We have shown that local surjectivity, which implies transversality to the level sets of any smooth function, only fails for a null set of quantum systems.
This technique can be applied to show that transversality of $V_T$ to the level sets of a specific objective function is also generic.
Considering a family of systems parameterized by a vector of numbers $\lambda$ (for example: the coupling constants of a spin chain system controlled by a magnetic field). 
Then, by a nearly identical argument, if $V_T(\lambda)[E]$ is transverse to a given level set $L$ of fidelity $J$ when variations of both $E$ and $\lambda$ are admissible, then all but a null set of fixed values of $\lambda$ leave $V_T[E]$ transverse to the same set when only variations of $E$ are considered.
We further note that the argument of the central theorem does not rest on the initial appeal to a space of piecewise constant controls.
The assumption of a finite dimensional space of controls is  physically unrestrictive as the space of controls can be given a very high dimension leaving the space of control fields still including all those which can be physically implemented.
One way to achieve this truncation of control space, is to only consider control fields with `bounded variation' (band limited in the terminology of Fourier series).
Such control fields possess no frequencies above a certain critical frequency.
If this critical frequency is high enough, then any discarded component of the control field would not affect the time evolution differently from noise.
An identical argument could be made if one started with any finite dimensional space of controls which can be shown to render the end-point map locally surjective.

The following table \ref{cortab} clarifies the correspondence between the abstract statement of the parametric transversality theorem (PTT) and its application to quantum control.

\begin{table}[h!]
\centering
\captionsetup{justification=centering}
\caption{Correspondence of objects for the quantum control application of the parametric transversality theorem}
\label{cortab}
\begin{tabular}{c | C{0.7\textwidth}}
Abstract PTT & Quantum Control Application \\ \hline
$M$ & Control Space \\ \hline
$N$ & $SU(n)$ \\ \hline
$S$ & Space of Fixed Values of Control \\ \hline
$L$ & Level set of $J$ (transversality of $V_T$), \newline or whole of $SU(n)$ (local surjevtivity) \\ \hline
$\phi(p,s)$  & $V_T[\lambda, E]$  \\ \hline
$\psi_s(p)$  & $V_T(\lambda)[E]$  (also written as $V_T(a_{j', k'})[a_{j,k}]$)
\end{tabular}
\end{table}

\subsection{Context and Physical Relevance}

The result of section (\ref{centhmsec}) is a parallel of the central result of Altifini \cite{alt}.
Altifiti's key result can be restated as: for \emph{almost all} $(a,b) \in \mathfrak{su}(n) \times \mathfrak{su}(n)$ (i.e., all but a null set) the map $V_T$ is \emph{globally} surjective (onto) on $SU(n)$.
Clearly, this result includes cases with more than one control field as it applies to the case with only a single control field and adding additional control fields cannot destroy controllability.
In the same spirit, the conclusion of section \ref{centhmsec} is that \emph{almost all} Hamiltonians will not produce singular traps.

It has been discussed in \cite{mepol} that while many practical control scenarios are described by a Hamiltonian in the dipole approximation, laboratory scenarios can include more complex forms of coupling.
As such, the mathematical existence of traps in systems in the dipole approximation, does not necessarily translate into the existence of traps in laboratory experiments.
Complex (non-linear) coupling of control fields within the system's Hamiltonian and coupling to the environment can be present in reality, even if they are only very small. 
Noise in the control field is also present in practice.
If fixed values of complex coupling to external fields are chosen at random, with certainty traps will not be present in the landscape.
This implies that, if a well validated model of a quantum system is shown to possess traps, even the least inaccuracy in the model or the presence of any additional types of coupling to the control fields, will very likely remove these traps by effectively perturbing the Hamiltonian matrix elements out of the failing set.

\subsection{Numerical Confirmation}

In order to confirm the conclusions, of the central Theorem \ref{tcthm}, numerical simulations were run.
$100$ pairs $a,b \in \mathfrak{su}(4)$ (i.e., four level systems) were generated at random and $100$ optimizations with random initial controls $E$ were run using a gradient ascent algorithm.
The control fields tested were defined to be piecewise constant with $1000$ pieces with appropriately bounded magnitude, as per the central result above.
Initially, a singular control $E$ was generated at random (see appendix \ref{singcon}) by choosing $B$ at random (i.e., a procedure specifically seeking a singular control, rather than one designed to avoid them).
The bounded magnitude premise of the central theorem was imposed via rejection sampling on $B$.
We will denote by $E[B]$ the control created from formula $\ref{singcon}$ with a given value of $B \in \mathfrak{su}(n)$ (with $B$ normalized s.t. $||B|| = 1$).
The time evolution simulation is given by:
\begin{align}
\frac{dU_t}{dt} = \left(a + E(t)b\right)U_t
\end{align}
Stochastic gradient ascent (over $B$) was then run to maximize the quantity:
\begin{align}
\left\langle B, \ U_T^{\dagger} \nabla J \big|_{U_T} \right\rangle
\end{align}
One readily confirms that this quantity will be maximized for a singular critical point and at no other point.
It was found that in all generated cases $a,b$, no singular critical controls consistent with the magnitude bound premise of Theorem \ref{tcthm} were discovered as the algorithm did not converge.
As many initial conditions were tested, this is strong numerical confirmation that such controls do not exist.

\section{Discussion and Conclusions}

Theorem \ref{tcthm} has been presented describing the structure of typical quantum control landscapes by introducing a novel tool from differential geometry extracted from the parameteric transversality theorem.
The technique used to obtain Theorem \ref{tcthm} is novel, and potentially has scope well beyond quantum control.
We have shown that quantum systems with singular critical controls are rare in the space of all possible quantum control systems, i.e., all systems evolving under the Schr\"odinger equation with some coupling to external control fields.
In order for the transversality result to apply specifically to the dipole approximation with a single field it would be required to show that the restrictions on the maximally controllable Hamiltonian do not leave the remaining Hamiltonian within the null set possessing traps.
However, as the set possessing traps is null, there exists a perturbation to any controllable system possessing traps, which removes them.
We have further formalized a sufficient condition for systems with singular critical points to possess no second order critical points (\ref{app2}).
Attempting to show that this condition almost always holds will be the focus of further work also based on the parametric transversality theorem.

We have refined the surjectivity assumption of quantum control landscape theory to that of transversality.
We argued that the end point map being transverse to all level sets of fidelity is sufficient for the dynamical and kinematical landscapes to share the same critical point structure.
A novel technique for showing that a very large class of realistic systems possess this property was also presented.

\label{axioms}

The original and current status of the three assumptions of quantum control landscape analysis are given here, so that the new and original forms and statues of each can be compared.

\begin{table}[h]
\centering
\resizebox{\textwidth}{!}{%
\begin{tabular}{l|l|l}
Original Assumptions & Current Assumptions & Status of Assumption \\ \hline
\begin{tabular}[c]{@{}l@{}} 1) Controllability:\\ $V_T$ is globally surjective\\ on the space of all controls fields \end{tabular} & \begin{tabular}[c]{@{}l@{}} 1) Controllability:\\ $V_T$ is globally surjective \\ on the space of all controls fields \end{tabular} & \begin{tabular}[c]{@{}l@{}}Shown to hold for\\  all but a null set of Hamiltonians \\ \cite{alt}, \cite{jur} (ch. 6.4, theorem 12, pg 188.) \end{tabular} \\ \hline
\begin{tabular}[c]{@{}l@{}} \textbf{2 a)} Local controllability:\\ $V_T$ is locally surjective.\\ i.e., $\frac{\delta V_T}{\delta E}$ is non-singular\end{tabular} & \begin{tabular}[c]{@{}l@{}} \textbf{2 b)} Transversality:\\ $V_T \pitchfork L_k$\\ for all level sets $L_K$ of $J$\end{tabular} & \begin{tabular}[c]{@{}l@{}}Both shown to hold for\\  all but a null set of\\ Hamiltonians\end{tabular} \\ \hline
\begin{tabular}[c]{@{}l@{}} 3) Resources: \\ Control resources sufficient \\ for $V_T$ to be globally surjective\end{tabular} & \begin{tabular}[c]{@{}l@{}} 3) Resources: \\ Control resources sufficient for $V_T$ \\ to be globally surjective \end{tabular} & Scenario dependent. \\
\end{tabular}
}
\caption{The three assumptions of quantum control landscape analysis expressed in their original and current forms. Assumption (2) has been significantly relaxed, although it is almost always satisfied for the original and current assumptions. The relaxation of the current assumption (2) greatly expands the set of Hamiltonians with no singular critical points.}
\label{axiomstab}
\label{tabas}
\end{table}
\FloatBarrier
It is possible to check that these assumptions imply that the dynamical landscape almost always possess the same critical point structure as the kinematical one.
This can be achieved by observing that the composition of two functions, $J(V_T[w])$, has the same critical point structure as that of $J$ if $V_T$ satisfies both transversality and globally surjectivity.
Finally, we note in the current assumptions of Table \ref{tabas} that assumptions 1 and 2, are \emph{almost always} satisfied.
As a result, it is clear that the primary determining factor for the ease for optimization in quantum control, for the vast majority of Hamiltonians, is the availability of sufficient control resources, assumption (3).

We also note that transversality only requires that the end-point map $V_T$ has rank 1 and the $\nabla J$ is contained within its range.
This is in contrast to the requirement that $V_T$ is \emph{full} rank, because full means that the rank is equal to the dimension of the group $SU(n)$, which is $n^2-1$.
The condition that local surjectivity be satisfied becomes more demanding on the rank of $V_T$ as the number of level rises, which is in contrast to transversality.

\section{Acknowledgments}

The authors would like to thank Robert Kosut for many helpful discussions and comments during the drafting of this work, and Daniel Burgarth for some corrections.
Benjamin Russell acknowledges the support of DOE (grant no. DE-FG02-02ER15344.
Herschel Rabitz acknowledges the support of the Army Research Office (Grant No. W911NF-16-1-0014).
We also acknowledge partial support from the Templeton foundation.

\section*{Bibliography}

\bibliographystyle{unsrt}
\bibliography{mybib}

\appendix
\section{Singular Controls in the Dipole Approximation and Beyond}

\label{singapp}

Here we give formulas for singular controls in the case of a single control field in the dipole approximation.
In order to express all relevant quantities in terms of Lie algebraic objects in $\mathfrak{su}(n)$ we make the following definitions:
\begin{align}
	\label{laf}
	& a := -i H_0 \\ \nonumber
	& b := -i H_c
\end{align}
Much of what is presented in this section and throughout is applicable, appropriately modified, to control systems of the analogous form on compact, connected, semi-simple Lie groups other than $SU(n)$.

In this notation, the Schr\"odinger equation for a controlled quantum system (\ref{schcon}) reads:
\begin{align}
	\label{schcons}
	\frac{d U(t)}{dt} = \left(a + E(t)b\right)U(t)
\end{align}
The authors are only aware of singular control trapping studies in the dipole approximation expressed here other than in a recent work \cite{mepol}.

It is possible to explicitly express $\delta V_T$ in terms of $\delta E$ \cite{in1};
\begin{align}
\delta V_T[w] = U(T) \int_{0}^{T} \delta E(t) U(t)^{\dagger} b U(t) dt
\end{align}
The right translation of the $U(T)^{\dagger} \delta U(T) \in \mathfrak{su}(n)$ is given by:
\begin{align}
\label{fm1}
U(T)^{\dagger} \delta U(T) & = \int_{0}^{T} \delta E(t) U(t)^{\dagger} b U(t) dt \\
& = \int_{0}^{T} \delta E(t) b_t dt \nonumber
\end{align}
where we have defined $b_t = U(t)^{\dagger} b U(t) = Ad_{U(t)}(b)$ so that it solves the adjoint equation $\frac{d b_t}{dt} = [b, H_t]$.
We note that this form (i.e., the integral of an adjoint orbit in a Lie algebra) is only possible because we have a Lie group of time evolution operators and that form is unique to such a scenario.

If the control $E$, which drove the system along the trajectory $U(t)$, is singular, then there exists, be definition, at least one $B U(T) \in T_{U(T)}SU(n)$ such that:
\begin{align}
	\label{fm2}
	\la B U(T,0), \delta U(T,0) \ra_{U(T,0)} = 0 \ \ \forall \delta E
\end{align}
where $\la, \ra_{U(T)}$ is any inner product on the tangent space $T_{U(T)}SU(n)$ and $B \in \mathfrak{su}(n)$.
A convenient choice of inner product is given by the unique (up to a constant multiple) bi-invariant Riemannian metric on $SU(n)$.
This metric is expressed as the right (or left, as the two coincide) translation of the Killing form $K: \mathfrak{su}(n) \times \mathfrak{su}(n) \rightarrow \mathbb{R}$:
\begin{align}
	K(X,Y) = \Tr(X^{\dagger}Y)
\end{align}
The condition that a control $E$ is singular can now be written, by applying formulas (\ref{fm1}, \ref{fm2}), in terms of this inner product:
\begin{align}
K \lb U(T)^{\dagger} \delta U(T), B \rb = 0 \Rightarrow \\ \nonumber
K \lb \int_{0}^{T} \delta E(t) b_t dt, B \rb = 0 \Rightarrow \\ \nonumber
\int_{0}^{T} \delta E(t) K \lb b_t, B \rb dt = 0
\end{align}
In scenarios where the set of $\delta E$ considered is large enough, one can apply the fundamental lemma of the calculus of variations and conclude that:
\begin{align}
	\label{fund}
	K \lb b_t, B \rb = 0 \ \ \forall t \in [0,T]
\end{align}
The interpretation of this equation is that infinitesimally varying $E(t)$ doesn't in turn vary $U(T,0)$ in the direction $B$ for any time during the evolution.
From this result, the form of the singular controls can be deduced by differentiating w.r.t $t$.
\begin{align}
	\frac{d}{dt} K \lb Ad_{U(t)}(b), B \rb & = 0 \ \ \forall t \in [0,T] \Rightarrow \\ \nonumber
	K \lb Ad_{U(t)}(ad_{b}(a+E(t)b)), B \rb & = 0 \ \ \forall t \in [0,T] \Rightarrow \\ \nonumber
K \lb Ad_{U(t)}(ad_{b}(a)), B \rb & = 0 \ \ \forall t \in [0,T]
\end{align}
Noting that $E$ has dropped out as $ad_b(b)=0$, one differentiates again w.r.t $t$ to find: 
\begin{align}
\label{singcon}
E(t) & = -\frac{K\lb Ad_{U(t)}(ad_{a}(ad_b(a))), B\rb}{K\lb Ad_{U(t)}(ad_{b}(ad_b(a))), B\rb} 
\end{align}
assuming that the denominator is never zero.
If it is zero, further differentiation is required which yields another expression containing higher levels of nested commutators.
Another formula can also be obtained by applying a symmetry of the Killing form to obtain from (\ref{fund})
\begin{align}
	K \lb b, B_t \rb = 0 \ \ \forall t \in [0,T]
\end{align}
where $B_t$ is defined similarly to $b_t$ above.
Closely following \cite{Wu}, the form of the singular controls can be similarly determined in the case of controlling the density matrix.

From formula (\ref{singcon}), one sees that not all singular controls are constant.
One also observes that there is no reason, \emph{a priori}, to preclude the possibility of encountering singular controls during gradient ascent for systems in the dipole approximation.
However, formula (\ref{singcon}) also indicates that there are very few singular controls relative to the size of the total control space.
This can be deduced from the one-to-one correspondence between $B \in \mathfrak{su}(n)$, and singular controls.
Given that the dimension of $\mathfrak{su}(n)$ is $n^2-1$, and the control space typically will have far larger dimension for an $n$ level system, this indicates that singular controls are not prevalent in the control space.
If the Hamiltonian contains quadratic coupling to a single control field, a similar implicit formula for a singular control can be found \cite{mepol} by a like procedure.

\section{Second Order Analogue of the Central Theorem \ref{tcthm}}
\label{app2}

The parametric transversality theorem can be applied to the Hessian of the fidelity by defining a new map $Q_T: TC \rightarrow TSU(n)$:
\begin{align}
Q_T(E, \delta E) := (V_T(E), d V_T\big|_{E} (\delta E))
\end{align}
where $d V_T$ is the push-forward of the map $V_T$.
This approach is appropriate for showing that only a null set of systems, which possess singular critical points, possess \emph{second order} critical points.
Another definition is now needed:
\begin{definition}
	Given two vector spaces $K,L$ define $\Sigma := \{M \in \text{Lin}(K,L) \ \big| \ \text{rank}(M) < \text{dim}(L) \}$.
That is $\Sigma$ is the set of rank deficient linear maps from $K$ to $L$.
\end{definition}
It is now possible to state the condition that a given quantum system has no second order critical points in terms of transversality of the map $Q$ to a specific submanifold $L$ of the range of the same map.

If $S \subset SU(N)$ is the submanifold constructed from the union of all the level sets (of J) which contain a singular critical value of $V_T$, define a submanifold of $L \subset TSU(n)$ as:
\begin{align}
L = S\times \Sigma
\end{align}
The condition that there are no second order critical points now reads: $Q_T \pitchfork L$, which renders it amenable to assessment via application of the parametric transversality theorem.
Checking for which systems this condition holds, and similarly for higher order conditions, will form the basis of further work.

\section{Singularities of the Matrix Exponential}
\label{meapp}

The matrix exponential $\text{exp}: A \mapsto e^{A}$ possess singularities at, and only at, any matrix $A$ for which $ad_{A} = [A, \cdot]$ has $2\pi n i$ as an eigenvalue for any $n\in \mathbb{N}$.
In any quantum system with a piecewise constant Hamiltonian (with $K$ pieces of duration $\delta t$ labeled as $H^{(k)}$, which are not time dependent), the propagator $U_T$ can be expressed as a product of exponentials:
\begin{align}
U_T = e^{\delta t i H^{(K)}} \cdots e^{\delta t i H^{(1)}}
\end{align}
The eigenvalues of $ad_{A}(\cdot)$ are $a_{n,m} = a_n - a_m$ where $a_n$ are the eigenvalues of $A$ \cite{fulton} for any square matrix $A$.
If the difference of the highest and lowest magnitude eigenvalue of $\delta t i H^{(n)}$ (which equals the magnitude largest eigenvalue of $ad_{\delta t i H^{(n)}}$) is less that $2 \pi$, then the map $H^{(n)} \mapsto e^{\delta t i H^{(n)}}$ is locally surjective.
If all $H^{(n)}$ meet this premise, then clearly the overall end-point map is surjective.
If a given $H^{(n)}$ does not meet the premise, then reducing $\delta t$ sufficiently will restore the premise.
Specifically, systems such that:
\begin{align}
\delta t < \frac{2 \pi}{\big|E^{(n)}_{\text{max}} - E^{(n)}_{\text{min}} \big|}
\end{align}
where $E^{(n)}_{\text{max}}$ is the largest eigenval of $H^{(n)}$ (and $E_{\text{min}}$ similarly) for all $N\in \{1, K\}$, will have no singularities.

More generally, the Lie theoretic exponential map, $\exp(\cdot)$, has the property that it is invertible (i.e., a diffeomorphism) in a neighborhood of $0$, i.e., $d \exp$ is invertible near to zero (Corollary 3.44, \cite{hall2003}).
This indicates that it is free from singularities, and thus that for piecewise constant controls, with sufficiently small pieces, the end-point map is also free from singularities.

\subsection{The Relationship between Singular Trajectories when Controlling the Propagator and when Controlling a Quantum State}
\label{geocon}

It has been noted in \cite{Wu} that there exist controls for systems of the type (\ref{schcon}) which are singular for the control of the the propagator and not for the control a quantum state.
Here we discuss the geometric basis of this finding using the construction of complex protective space as the quantum state space of pure states is a homogeneous space of $SU(n)$.
We also discuss why this construction has no analogue in the case of a mixed state.

In geometric quantum mechanics the space of pure states, up to equivalence of states by normalization and global phase, is a manifold known as a complex projective space \cite{GQM}.
For other applications of this construction in quantum control see \cite{me0, me1, me2}.
This relationship can be formally expressed as:
\begin{align}
	\mathbb{C}P^{n-1} \cong SU(n)/U(n-1)
\end{align}
where the $/$ symbol refers to the quotient of $SU(n)$ into $U(n-1)$ cosets.
The special case of the space $\mathbb{C}P^1$, is the familiar Bloch sphere for a qubit.
This construction is only possible because $SU(n)$ acts transitively on $\mathbb{C}P^{(n-1)}$, see \cite{EP} for details of this construction.
As unitary evolution preserves the degree of mixedness of a mixed state, $SU(n)$ does not act transitively on the manifold of all density matrices representing mixed states.
Due to this lack of transitive action (i.e., not all states can be transformed into all others via unitary evolution), this space cannot be represented as a homogeneous space of $SU(n)$.
Because of this circumstance, it is not as straight forward to form an analogous construction for mixed states.

The space where the quantum state resides is a quotient of the space where the time evolution operator resides.
Because of this, each direction in the tangent space at some point in $\mathbb{C}P^{n-1}$ corresponds to more than one direction in the tangent space at some point in $SU(n)$.
As such, the existence of a direction $\delta V_T \in T_{U}SU(n)$, which does not correspond to any admissible variation $\delta E$, is not sufficient for such a direction to exist in $T \mathbb{C} P^{(n-1)}$ when the time evolution of the state is standardly defined by $|\psi_t \ra = U_t| \psi_0 \ra$.
Thus, not all controls of the form (\ref{singcon}) are singular for the control of a pure state.
This construction provides insight into which controls can be singular for the control of a pure state, as it reveals the exact sense of: several directions $\delta U_T$ in which $U_T$ and be steered by variation $\delta E$ of $E$ corresponding to a given variation $\delta |\psi_T \rangle$ of $|\psi_T\rangle$.
Thus, it is clear that not every control which is singular for the control of the propagator is singular for the control a pure state.
The practical significance of this result is that if the landscape for the control of the propagator $U_T$ is trap free, then the landscape for the control of the state $|\psi \rangle$ is also trap free.

\end{document}